\documentclass[runningheads,a4paper]{llncs}

\usepackage{amssymb}
\usepackage{graphicx}
\usepackage{times}
\usepackage{algorithmic}
\usepackage{algorithm}
\usepackage{epsfig}
\usepackage{amssymb}
\usepackage{amsmath}
\usepackage{amsfonts}
\usepackage{amsmath}
\usepackage{nccmath}
\usepackage{url}

\newcommand{\keywords}[1]{\par\addvspace\baselineskip
\noindent\keywordname\enspace\ignorespaces#1}

\begin{document}

\mainmatter

\title{Efficient Algorithms for Approximate Triangle Counting}

\titlerunning{Efficient Algorithms for Approximate Triangle Counting}

\author{Mostafa Haghir Chehreghani}

\authorrunning{Mostafa Haghir Chehreghani}

\institute{Department of Computer Science,\\
Katholieke Universiteit Leuven, Belgium\\
Mostafa.HaghirChehreghani@cs.kuleuven.be\\
}

\maketitle

\begin{abstract}
Counting the number of triangles in a graph has many important applications in network analysis.
Several frequently computed metrics like the clustering coefficient and the transitivity ratio
need to count the number of triangles in the network.
Furthermore, triangles are one of the most important graph classes considered in network mining.
In this paper, we present a new randomized algorithm for approximate triangle counting.
The algorithm can be adopted with different sampling methods and give effective triangle counting methods.
In particular, 
we present two sampling methods, called the \textit{$q$-optimal sampling} and the \textit{edge sampling},
which respectively give $O(sm)$ and $O(sn)$ time algorithms with nice error bounds
($m$ and $n$ are respectively the number of edges and vertices in the graph and $s$ is the number of samples).
Among others, we show, for example,
that if an upper bound $\widetilde{\Delta^e}$ is known for the number of triangles incident to every edge,
the proposed method provides an $1\pm \epsilon$ approximation
which runs in $O( \frac{\widetilde{\Delta^e} n \log n}{\widehat{\Delta^e} \epsilon^2}  )$ time,
where $\widehat{\Delta^e}$ is the average number of triangles incident to an edge.
Finally we show that the algorithm can be adopted with streams.
Then it, for example, will perform 2 passes over the data (if the size of the graph is known, otherwise it needs 3 passes)
and will use $O(sn)$ space.
\end{abstract}


\section{Introduction}
\label{section:introduction}

Graphs are fundamental structures for modeling complex relationships between data.
Examples include: Internet (where vertices are routers and edges correspond to physical links),
World Wide Web (where vertices are web pages and edges correspond to hyperlinks),
social networks (where vertices are humans and edges correspond to friendships),
traffic data (where vertices are places or cities and edges correspond to roads) and
biological networks (where vertices are proteins and edges correspond to protein interactions).

The problem of counting subgraphs of a certain class,
is one of the typical problems in graph mining which in recent years has obtained considerable attentions.
Triangles are one of the most important basic subgraphs.
On the other hand, computation of several network indices and statistics are based on counting the number of triangles, 
which makes triangle counting an essential problem in network analysis. 
\textit{Clustering coefficient} of a graph \cite{watts:jrnl} is defined as the normalized sum of the fraction of neighbor pairs
of a vertex of the graph that are connected.
\textit{Transitivity coefficient} of a graph \cite{harary:jrnl}, is defined as the ratio between
three times the number of triangles and the number of length two paths in the graph.

In terms of time complexity, the most efficient triangle counting algorithms are based on matrix multiplication.
If $A$ is the adjacency matrix of a graph $G$, the number of triangles in $G$ is equal to 
\begin{equation}
\label{eq:matrixmultiplication1}
\frac{1}{6} \mathbf{Tr} (A^3)  
\end{equation}
where $\mathbf{Tr}( \cdot )$ denotes the \textit{trace} of a matrix
defined as the sum of the elements on the main diagonal of the matrix.

Time complexity of the most efficient known algorithm for matrix multiplication is $O(n^{2.3727})$ \cite{coppersmith:jrnl}, \cite{williams:proceedings},
where $n$ is the number of vertices of the graph (the number of rows/columns of $A$).
The exponent of $n$, denoted by $\omega$, is called \textit{matrix multiplication exponent}.
Alon et al. give in \cite{alon:jrnl} a more efficient triangle counting algorithm for sparse graphs.
Time complexity of their method is $O(m^{\frac{2\omega}{\omega+1}})$, where $m$ is the number of edges of the graph.
For $\omega=2.3727$, this time complexity is equal to $O(m^{1.41})$.

However, exact triangle counting methods may be inefficient when the size of the graph is large.
In these cases, an approximate algorithm is preferred in the cost of losing the exact number of triangles.
In recent years, many algorithms have been proposed for approximate triangle counting.
A widely used technique is the \textit{sparsification technique} \cite{tsourakakis:jrnl}, \cite{tsourakakis:proceedings}, \cite{pagh:jrnl} and \cite{kolountzakis:jrnl}.
In this technique, the graph is converted into a sparse graph and 
the number of triangles in the sparsified graph is counted.
Then, the result is scaled to the original graph.
Some other methods are based on approximate computation of algebraic properties of the graph
like \textit{eigenvalues} of the adjacency matrix \cite{tsourakakis:icdm}.
However, well-known methods for computing (or approximating) eigenvalues are heavily based on matrix multiplication
and it is known that worst case time complexity of computing eigenvalues is the same as matrix multiplication.
On the hand, such approximate triangle counting algorithms look at the algebraic methods
(like approximate matrix multiplication and low rank matrix approximation) as a black-box. 
However, samplings in the algebraic methods are done in a way to minimize the element-wise error or
the frobenius norm of the error matrix
and therefore, the error of triangle counting is not minimized.

In this paper, we propose a new randomized algorithm for approximate triangle counting.
Our method is a variation of several approximate matrix multiplication algorithms
\cite{drineas:soda}, \cite{drineas2:jrnl} and \cite{drineas:jrnl}.
However, it does not generate any product matrix and several algorithmic aspects are different.
Furthermore, the sampling methods which are crucial elements of the algorithm, are also different.
The proposed algorithm can be seen as a general framework to which different sampling methods can be applied.
Every sampling method gives a new triangle counting algorithm with its own error bounds and time complexity. 
In particular, we present two sampling methods, called the \textit{$q$-optimal sampling} and the \textit{edge sampling},
which respectively give $O(sm)$ and $O(sn)$ time algorithms with nice error bounds
($m$ and $n$ are respectively the number of edges and vertices in the graph and $s$ is the number of samples).
Among others, we show, for example,
that if an upper bound $\widetilde{\Delta^e}$ is known for the number of triangles incident to every edge,
the proposed method provides an $1\pm \epsilon$ approximation
which runs in $O( \frac{\widetilde{\Delta^e} n \log n}{\widehat{\Delta^e} \epsilon^2}  )$ time,
where $\widehat{\Delta^e}$ is the average number of triangles incident to an edge.
As we will discuss, some existing algorithms can be seen as adoptations of the proposed algorithm with specific sampling methods.
We finally show that the algorithm can be extended to streams.
Then it, for example, will perform 2 passes over the data (if the size of the graph is known, otherwise it needs 3 passes)
and will use $O(sn)$ memory cells.  

The rest of this paper is organized as follows.
In Section \ref{section:approximatealgorithm}, we present a new randomized algorithm for approximate triangle counting.
In Section \ref{sec:samplings}, different sampling methods are introduced.
In Section \ref{sec:stream} we extend the algorithm for counting triangles in streams.
An overview of related work is provided in Section \ref{sec:relatedwork}.
Finally, the paper is concluded in Section \ref{sec:conclusions}.

Throughout the paper, $G$ refers to a simple (i.e. loop-free and without multiple edges) and undirected graph.  
$A$ refers to the adjacency matrix of $G$. Therefore, $A$ is a square matrix consisting of $0$s and $1$s.
$A_{ij}$ denotes the element in the $i$-th row and the $j$-th column of $A$.
$n$ and $m$ denote the number of vertices of $G$ (the number of rows and the number of columns of $A$)
and the number of edges of $G$, respectively. $\Delta$ denotes the number of triangles in $G$.
In this paper, we use an index $i$ for referring to a row (column) in the adjacency matrix as well as
for referring to the vertex of the graph corresponds to the row (column) $i$. 

\section{The approximate triangle counting algorithm}
\label{section:approximatealgorithm}

In this section, we present a randomized algorithm for approximate triangle counting.
Suppose $A$ is an $n \times n$ matrix which is the adjacency matrix of a graph $G$.
As Equation \ref{eq:matrixmultiplication1} shows, in order to count triangles in $G$, we can compute the trace of $A^3$.
Algorithm \ref{alg:approximate1} shows the high level pseudo code of an approximate triangle counting algorithm,
based on randomized calculation of the trace of $A^3$.

In every iteration $t$ of the loop in Lines \ref{alg:approximate1:loop1}-\ref{alg:approximate1:loop2} of Algorithm \ref{alg:approximate1}, 
first the following probabilities are computed:
\begin{equation*}
p_1,p_2, \hdots, p_n \geq 0 \text{ such that } \sum_{i=1}^n p_i=1  
\end{equation*}

Then, an $i \in \{ 1,\hdots, n\}$ is selected with probability $p_i$,
and the following probabilities are computed:
\begin{equation*}
q_{1|i}, \hdots, q_{n|i} \geq 0 \text{ such that } \sum_{j=1}^n q_{j|i}=1
\end{equation*}

Finally, an $j \in \{1,\hdots,n\}$ is selected with probability $q_{j|i}$ and the number of triangles is estimated by 
\begin{equation}
\label{eq:betat}
\beta_t=\frac{\sum_{d=1}^{n}A_{di}A_{ij}A_{jd}}{6p_i q_{j|i}} 
\end{equation}

The final estimation, $\beta$, is the average of the estimations of different trials.


\begin{algorithm}
\caption{High level pseudo code of the approximate triangle counting algorithm.}
\label{alg:approximate1}
\mbox{\textsc{TriangleCounter}}
\begin{algorithmic} [1]
\REQUIRE {A graph $G$}.
\ENSURE The approximate number of triangles in $G$.
\STATE \COMMENT{Let $A$ be the $n \times n$ adjacency matrix of $G$}
\STATE $ \beta \leftarrow 0$
\FOR {$t=1$ \textbf{to} $s$} \label{alg:approximate1:loop1}
\STATE Compute $p_1,\hdots,p_n$
\STATE Select $i \in \{ 1, \hdots, n\}$ with the probability $p_i$
\STATE Compute $q_{1|i},\hdots,q_{n|i}$
\STATE Select $j \in \{ 1, \hdots, n\}$ with the probability $q_{j|i}$
\STATE $\beta_t \leftarrow \frac{\sum_{d=1}^{n}A_{di}A_{ij}A_{jd}}{6p_i q_{j|i}}$ \label{alg:approximate1:line1}
\STATE $\beta \leftarrow \beta + \beta_t$
\ENDFOR \label{alg:approximate1:loop2}
\RETURN $\frac{\beta}{s} $
\end{algorithmic}
\end{algorithm}

In the rest of this section, we study some important properties of Algorithm \ref{alg:approximate1}.  

\begin{lemma}
\label{lemma:lemma1}
In Algorithm \ref{alg:approximate1}, for every $t \in \{1, \hdots, s\}$ we have:
\begin{equation}
\mathbf E\big(\beta_t\big)= \mathbf E\big( \beta \big)= \frac{\mathbf {Tr}(A^3)}{6}
\end{equation}
\end{lemma}

\begin{proof}
For every $t \in \{1, \hdots, s\}$, we have:
\begin{align}
\mathbf E(\beta_t) &= \sum_{i=1}^{n} \sum_{j=1}^{n} \left( \frac{\sum_{d=1}^{n} A_{di}A_{ij}A_{jd} }{6p_{i} q_{j|i}}p_ip_{j|i} \right) \nonumber\\
		   &= \frac{1}{6} \sum_{i=1}^{n} \sum_{j=1}^{n} \sum_{d=1}^{n} A_{di}A_{ij}A_{jd} \nonumber\\
		   &= \frac{1}{6} \mathbf{Tr}(A^3) \nonumber
\end{align}
On the other hand, $\beta_t$'s are independent random variables
and $\beta$ is the sum of $s$ independent random variables $\beta_1 \hdots \beta_s$ divided by $s$.
Therefore
\begin{equation*}
\mathbf E(\beta) = \frac{\mathbf E (\sum_{t=1}^s \beta_t) }{s}= \frac{ s \mathbf E(\beta_t)}{s}= \mathbf E (\beta_t)= \frac{\mathbf {Tr}(A^3)}{6}
\end{equation*}
$\blacksquare$
\end{proof}

\begin{lemma}
In Algorithm \ref{alg:approximate1}, variance of every random variable $\beta_t$ is
\begin{equation}
\label{eq:variancet}
\mathbf{Var} (\beta_t) =
\frac{1}{36} \left(  \sum_{i=1}^{n} \sum_{j=1}^{n} \frac{ \big( \sum_{d=1}^{n} A_{di}A_{ij}A_{jd} \big)^2 }{ p_i q_{j|i}} - \big( \mathbf{Tr}(A^3) \big)^2  \right) 
\end{equation}
\end{lemma}

\begin{proof}
We have:
\begin{equation*}
\mathbf{Var}(\beta_t) = \mathbf E(\beta_t^2) - \big(\mathbf E(\beta_t) \big)^2
\end{equation*}
Then
\begin{align}
\mathbf E(\beta_t^2)= \sum_{i=1}^{n} \sum_{j=1}^{n} \big( \frac{\sum_{d=1}^{n} A_{di}A_{ij}A_{jd} }{6p_{i} q_{j|i}} \big)^2  p_i p_{j|i} = \frac{1}{36} \sum_{i=1}^{n} \sum_{j=1}^{n} \frac{ \big( \sum_{d=1}^{n} A_{di}A_{ij}A_{jd} \big)^2 }{ p_i q_{j|i}} \label{eq:var1}
\end{align}
and
\begin{align}
\big(\mathbf E(\beta_t) \big)^2 = \frac{1}{36} \big( \mathbf{Tr}(A^3) \big)^2 \label{eq:var2}
\end{align}

Therefore
\begin{equation}
\mathbf {Var}(\beta_t) = 
\frac{1}{36} \left(  \sum_{i=1}^{n} \sum_{j=1}^{n} \frac{ \big( \sum_{d=1}^{n} A_{di}A_{ij}A_{jd} \big)^2 }{ p_i q_{j|i}} - \big( \mathbf{Tr}(A^3) \big)^2  \right) 
\end{equation} 

$\blacksquare$
\end{proof}

Since $\beta$ is the average of $s$ independent copies of $\beta_t$, then
\begin{equation}
\label{eq:variance}
\mathbf {Var}(\beta) =  \frac{\mathbf {Var}(\beta_t)}{s}  =  \frac{1}{36s} \left(  \sum_{i=1}^{n} \sum_{j=1}^{n} \frac{ \big( \sum_{d=1}^{n} A_{di}A_{ij}A_{jd} \big)^2 }{ p_i q_{j|i}} - \big( \mathbf{Tr}(A^3) \big)^2  \right)
\end{equation}

For a vertex $i$ of the graph, \textit{local triangles} of $i$
are triangles which are incident to $i$.
The number of local triangles of $i$, denoted by $\Delta_i$, equals to
\[\Delta_i = \frac{1}{2}\sum_{j=1}^n \sum_{d=1}^n  A_{di}A_{ij}A_{jd}\]
Let $\{i,j\}$ be an edge of the graph.
\textit{Local triangles} of $\{i,j\}$ are triangles for which $\{i,j\}$ is an edge.
The number of local triangles of $\{i,j\}$, denoted by $\Delta_{\{i,j\}}$,
is equal to
\[\Delta_{\{i,j\}} = \sum_{d=1}^n A_{di}A_{ij}A_{jd}\]
If $i$ is not connected to $j$, the number of local triangles of $\{i,j\}$ is equal to $0$.
The following holds between $\Delta_i$ and $\Delta_{\{i,j\}}$:
$\Delta_i = \frac{1}{2} \sum_{j=1}^n \Delta_{\{i,j\}}$.
Let $\Delta$ refer to the number of triangles in the graph.
We have: $\Delta= \frac{1}{3} \sum_{i=1}^n \Delta_i$.

Using the notion of local triangles of edges, Equation \ref{eq:variance} can be re-written as:
\begin{equation}
\label{eq:variance1}
\mathbf {Var}(\beta) =  \frac{1}{36s} \sum_{i=1}^{n} \sum_{j=1}^{n} \frac{ {\Delta_{\{i,j\}}}^2 }{ p_i q_{j|i}} -\frac{\Delta^2}{s}  
\end{equation}

\begin{lemma}
\label{theorem:timecomplexity}
If in Algorithm \ref{alg:approximate1} probabilities $p_1,p_2, \hdots, p_n$
and $q_{1|i},q_{2|i}, \hdots, q_{n|i}$ are accessible in constant time, 
its time complexity will be $O(sn)$.
\end{lemma}
\begin{proof}
The loop in Lines \ref{alg:approximate1:loop1}-\ref{alg:approximate1:loop2} is performed for $s$ times.
Inside the loop, the most time consuming step is Line \ref{alg:approximate1:line1} which takes $O(n)$ time.
Therefore, time complexity of the algorithm is $O(sn)$. $\blacksquare$
\end{proof}

%
%
%

\section{Sampling methods}
\label{sec:samplings}

In this section, we present a number of sampling methods.
First in Section \ref{sec:optimalsampling} the optimal sampling is investigated.
Then, since the optimal sampling might be computationally expensive, other near-optimal samplings are introduced. 

\subsection{Optimal sampling}
\label{sec:optimalsampling}

Lemma \ref{theorem:optimalsamplingvariance}
introduces the probabilities and error bound of the optimal sampling.

\begin{lemma}
\label{theorem:optimalsamplingvariance}
If for $1 \leq i \leq n$, $p_i$'s are equal to 
\begin{equation}
\label{eq:pioptimal}
p_i=\frac{ \Delta_i }{3 \Delta }
\end{equation}
and then after selecting $i$, for $1 \leq j \leq n$, $q_{j|i}$'s are
\begin{equation}
q_{j|i}=\frac{\Delta_{\{i,j\}}}{2\Delta_i} 
\end{equation}
the variance of $\beta$ presented in Equation \ref{eq:variance} is minimized.
The minimized variance is $0$.
\end{lemma}

Since the error bound of the optimal sampling is $0$, it gives an exact triangle counting algorithm.
On the other hand, time complexity of computation of $p_i$'s in Equation \ref{eq:pioptimal}
is the same as time complexity of exact triangle counting.
Therefore, using Algorithm \ref{alg:approximate1} with the optimal sampling
is the same (in both accuracy and complexity) as using an exact triangle counting algorithm.

\subsection{$q$-optimal sampling}

In the $q$-optimal sampling, 
every vertex $i$, $1 \leq i \leq n $, is selected by some strategy (which can be performed in $O(1)$ time).
For example, they are selected uniformly at random (therefore $p_i=\frac{1}{n}$),
or they are selected proportional to their degrees
(therefore, $p_i=\frac{\mathbf{deg}(i)}{2m}$, where $\mathbf{deg}(i)$ refers to the degree of $i$).
Then, every vertex $j$ is selected in a way to minimize $\mathbf{Var} (\beta)$.



\begin{lemma}
\label{lemma:psemiuniform}
In the $q$-optimal sampling,
after choosing a vertex $i$,
if every vertex $j$ is selected with probability 
\begin{equation}
\label{eq:qoptimalprobability}
q_{j|i}=\frac{\Delta_{\{i,j\}} } { 2\Delta_i } 
\end{equation}
the variance of $\beta$ is minimized.
\end{lemma}

\begin{proof}

In order to minimize $\mathbf{Var}(\beta)$, we need to minimize $\sum_{j=1}^n \frac{{\Delta_{\{i,j\}}}^2}{q_{j|i}}$, because other parts of $\mathbf{Var}(\beta)$ are independent of $j$.
We define 
\[f(q_{1|i},q_{2|i}, \hdots , q_{n|i} ) = \sum_{j=1}^n \frac{{\Delta_{\{i,j\}}}^2}{q_{j|i}} \]
and substitute $q_{n|i}$ by $1-\sum_{j'=1}^{n-1} q_{j'|n}$ and form equations
$\frac{\partial f}{\partial q_{j|i}}=0$, for $1 \leq j \leq n-1$.

We get
\begin{align}
            \frac{{\Delta_{\{i,j\}}}^2}{ {q_{j|i}}^2 } = \frac{{\Delta_{\{i,n\}}}^2}{ \left( 1 - \sum_{j'=1}^{n-1} q_{j|i} \right)^2 } \nonumber \\
\Rightarrow q_{j|i} =  \frac{ 1 - \sum_{j'=1}^{n-1} q_{j|i} }{ \Delta_{\{i,n\}} }  \Delta_{\{i,j\}}  \label{eq:qoptimalproof1} 
\end{align}

Summing $q_{j|i}$'s, for $1 \leq j \leq n-1$, and doing simplifications, we get
\begin{equation*}
\sum_{j=1}^{n-1} q_{j|i} =  \frac{ \sum_{j'=1}^{n-1} {\Delta_{\{i,j\}}} }{ \sum_{j'=1}^{n} {\Delta_{\{i,j\}}} }
\end{equation*}

Putting the value of $\sum_{j=1}^{n-1} q_{j|i} $ into Equation \ref{eq:qoptimalproof1} and doing simplifications,
we get the value of $q_{j|i}$ for the $q$-optimal sampling:
\begin{equation*}
q_{j|i} = \frac{ \Delta_{\{ i,j\}} }{ \sum_{j'=1}^n \Delta_{\{ i,j'\}} } = \frac{ \Delta_{\{ i,j\}} }{2 \Delta_i} 
\end{equation*}
$\blacksquare$
\end{proof}

If vertices $i$ are selected proportional to their degrees, the variance of $\beta$ in the $q$-optimal sampling will be 
\begin{equation}
\label{eq:psemivariancedegree}
\mathbf{Var}(\beta)=
\frac{2m}{9s} \sum_{i=1}^n \frac{{\Delta_i}^2}{\mathbf{deg}(i)}  - \frac{\Delta^2}{s}
\end{equation}

and if they are selected selected uniformly at random, $\mathbf{Var}(\beta)$ will be
\begin{equation}
\label{eq:psemivarianceuniform}
\mathbf{Var}(\beta)=
\frac{n}{9s} \sum_{i=1}^n {\Delta_i}^2 - \frac{\Delta^2}{s}
\end{equation}

The motivation for selecting vertices $i$ proportional to their degrees is that,
as studied in \cite{tsourakakis:icdm},
in most of real-world networks,
vertices of higher degrees have higher number of local triangles.
Then, for every two vertices $i$ and $i'$, if it holds that $deg(i)\geq deg(i')$ implies $\Delta_i \geq \Delta_{i'}$,
it can be shown that selecting vertices $i$ proportional to their degrees gives a better sampling than choosing them uniformly at random.

If the $q$-optimal sampling is used,
in every iteration of the loop in Lines \ref{alg:approximate1:loop1}-\ref{alg:approximate1:loop2} of Algorithm \ref{alg:approximate1},
probabilities $p_i$ and $q_{j|i}$ can be computed in $O(m)$ time.
On the other hand, it takes $O(n)$ time to compute $\beta_t$.
Therefore, time complexity of Algorithm \ref{alg:approximate1} with the $q$-optimal sampling will be $O(sm)$.

The $q$-optimal sampling can provide efficient $1 \pm \epsilon$ approximations, specifically if some information on local triangles of \textit{vertices} is available. 
For example, consider the version of the $q$-optimal sampling where vertices $i$ are selected uniformly at random ($p_i=\frac{1}{n}$).
Suppose that there exists an already known value $\widetilde{\Delta^v}$ such that for every vertex $i$, $\Delta_i\leq \widetilde{\Delta^v}$.
Then, in every iteration of the loop in Lines \ref{alg:approximate1:loop1}-\ref{alg:approximate1:loop2} of  Algorithm \ref{alg:approximate1},
a random variable $X_t$ can be defined as $ X_t = \frac{\beta_t}{n\widetilde{\Delta^v}} =\frac{\Delta_{i(t)} }{3\widetilde{\Delta^v}}$,
where $\Delta_{i(t)}$ is the number of local triangles of the vertex selected in iteration $t$.
We have: $\mathbf E(X_t) = \frac{\Delta}{n\widetilde{\Delta^v}}=\frac{\widehat{\Delta^v}}{\widetilde{\Delta^v}}$,
where $\widehat{\Delta^v}$ is the average number of local triangles of vertices.
By Chernoff bound we obtain:
\begin{equation}
\mathbf{Pr}\left[ \frac{1}{s} \sum_{t=1}^s X_t - \frac{\widehat{\Delta^v}}{\widetilde{\Delta^v}} > \epsilon  \frac{\widehat{\Delta^v}}{\widetilde{\Delta^v}} \right]
\leq \exp \left( -\frac{\epsilon^2 s\widehat{\Delta^v}}{2\widetilde{\Delta^v}}  \right) 
\end{equation}
If $s= \varOmega(\frac{\widetilde{\Delta^v} \log n}{\widehat{\Delta^v} \epsilon^2})$,
then $\frac{n\widetilde{\Delta^v}}{s} \sum_{t=1}^s X_t$ approximates $\Delta$ within a factor of $\epsilon$ with probability at least $1-n^{-c}$ for any constant $c$.
This gives an $O( \frac{\widetilde{\Delta^v} m \log n}{\widehat{\Delta^v} \epsilon^2}  )$ time algorithm which approximates $\Delta$ within a factor of $\epsilon$.
Specifically, if $\widetilde{\Delta^v}$ is greater than $\widehat{\Delta^v}$ only by a factor of a constant,
time complexity of the algorithm will be $O(\frac{m \log n}{\epsilon^2})$.

\subsection{Edge sampling}

In the edge sampling, first
a vertex $i$ is selected by some strategy (which can be done in $O(1)$ time).
Then, a neighbor $j$ of $i$ is selected by some (probably different) strategy which also can be done in $O(1)$ time.
Since in this sampling computation of probabilities is done in $O(1)$ time, time complexity of the algorithm is $O(sn)$. 

For example, $i$ can be selected uniformly at random (therefore, $p_i=\frac{1}{n}$), then,
for every $j$, if $j$ is a neighbor of $i$, $q_{j|i}$ is equal to $\frac{1}{\mathbf{deg}(i)}$; otherwise it is $0$.
This case is similar to the methods which uniformly sample an edge and count the number of triangles incident to it and scale the result.
The first algorithm proposed in \cite{jowhari:proceedings} and partially the algorithm of \cite{buriol:proceedings}
are examples of such methods.
In this case, variance of $\beta$ will be
\begin{align}
\label{eq:uniformvariance2}
\mathbf{Var} (\beta)= &  \frac{n}{36s} \sum_{i=1}^n \left( \mathbf{deg}(i)\sum_{j=1}^n  {\Delta_{\{i,j\}}} ^2 \right) - \frac{\Delta^2}{s}
\end{align}

In the second case of the edge sampling, $i$ is chosen with probability $p_i=\frac{\mathbf{deg}(i)}{2m}$.
Then, similar to the first case, for every vertex $j$, if $j$ is a neighbor of $i$, $q_{j|i}$ will be $\frac{1}{\mathbf{deg}(i)}$.
Otherwise, it will be $0$.
Variance of $\beta$ in this case is:
\begin{equation}
\label{eq:degreevariance2}
\mathbf{Var} (\beta) = \frac{m}{18s} \sum_{i=1}^n  \sum_{j=1}^n  {\Delta_{\{i,j\}}} ^2 - \frac{\Delta^2}{s}
\end{equation}

Similar to the $q$-optimal sampling, the edge sampling can provide efficient $1 \pm \epsilon$ approximations,
specifically if some information on local triangles of \textit{edges} is available.
For example, consider the second case and suppose that there exists a known value $\widetilde{\Delta^e}$ such that for every edge $\{i,j\}$, $\Delta_{\{i,j\}}\leq \widetilde{\Delta^e}$.
Then, in every iteration of the loop in Lines \ref{alg:approximate1:loop1}-\ref{alg:approximate1:loop2} of  Algorithm \ref{alg:approximate1},
a random variable $X_t$ can be defined as $ X_t = \frac{\beta_t}{m\widetilde{\Delta^e}} =\frac{\Delta_{\{i,j\}(t)} }{3\widetilde{\Delta^e}}$,
where $\Delta_{\{i,j\}(t)}$ is the number of local triangles of the edge $\{i,j\}$ selected in iteration $t$.
We have: $\mathbf E(X_t) = \frac{\Delta}{m\widetilde{\Delta^e}}=\frac{\widehat{\Delta^e}}{\widetilde{\Delta^e}}$,
where $\widehat{\Delta^e}$ is the average number of local triangles of edges.
Similarly, by Chernoff bound we obtain:
\begin{equation}
\mathbf{Pr}\left[ \frac{1}{s} \sum_{t=1}^s X_t - \frac{\widehat{\Delta^e}}{\widetilde{\Delta^e}} > \epsilon  \frac{\widehat{\Delta^e}}{\widetilde{\Delta^e}} \right]
\leq \exp \left( -\frac{\epsilon^2 s\widehat{\Delta^e}}{2\widetilde{\Delta^e}}  \right) 
\end{equation}
and if $s= \varOmega(\frac{\widetilde{\Delta^e} \log n}{\widehat{\Delta^e} \epsilon^2})$,
then $\frac{m\widetilde{\Delta^e}}{s} \sum_{t=1}^s X_t$ approximates $\Delta$ within a factor of $\epsilon$ with probability at least $1-n^{-c}$ for any constant $c$.
This gives an $O( \frac{\widetilde{\Delta^e} n \log n}{\widehat{\Delta^e} \epsilon^2}  )$ time algorithm which approximates $\Delta$ within a factor of $\epsilon$.
Specifically, if $\widetilde{\Delta^e}$ is greater than $\widehat{\Delta^e}$ only by a factor of a constant,
time complexity of the algorithm will be $O(\frac{n \log n}{\epsilon^2})$.

\section{Triangle counting in streams}
\label{sec:stream}

In many applications like World Wide Web and social networks,
the dataset is too large to load it into the main memory.
A widely used approach to address this problem is the use of the data stream model.
A data stream is an ordered sequence in which
data arrives one item at a time,
and the algorithm has access to limited computation and storage capabilities. 

In this section, we extend the algorithm presented in Section \ref{section:approximatealgorithm} to streams.
While our focus is the $q$-optimal sampling where vertices $i$ are selected uniformly at random,
the other sampling methods can be extended to streams in a similar way. Due to lack of space, we here omit details.
If the number of vertices of the graph, $n$, is already known, our algorithm will need 2 passes over the data.  
Otherwise, an extra pass will be needed to find it.
First, $s$ integers (vertices) $i$ are selected independently at random with uniform probability $\frac{1}{n}$.
Then:
\begin{itemize}
\item
During the first pass, for every $i$, the \textit{neighborhood vector} of $i$, denoted by $\mathcal I^i$, is formed. 
The neighborhood vector of $i$, shows which vertices of the graph are a neighbor of $i$ and which ones are not.
For every vertex $j$, if $\{i,j\}$ is an edge of the graph, ${\mathcal I^i}_j$ is set to $1$, otherwise, it is set to $0$.  
\item
During the second pass, for every vertex $i$ and for all vertices $j$,
the number of local triangles of the edge between $i$ and $j$ is calculated.
To do so, for every $i$ a vector $\mathcal P^i$ of size $n$ is used,
where ${\mathcal P^i}_j$ stores the number of local triangles of $\{i,j\}$.
When an edge $A_{jd}$ is visited, 
if there exists an edge between $i$ and $j$ (i.e. ${\mathcal I^i}_j=1$) and an edge between $i$ and $d$ (i.e. ${\mathcal I^i}_d=1$),
a triangle consisting of the vertices $i$, $j$ and $d$ is found.
Therefore, ${\mathcal P^i}_j$ and ${\mathcal P^i}_d$ and $z^i$ are increased by $1$.
$z^i$ is used to store the number of local triangles of $i$.
\end{itemize}

During these two passes, the information required for $q$-optimal sampling and approximate triangle counting are gathered.
Then, for every $i$,
a $j \in \{ 1,\hdots,n\}$ is selected with probability $q_{j|i}=\frac{\Delta_{\{i,j\}}}{2\Delta_i} = \frac{\mathcal P_j}{2z^i} $
and the approximate number of triangles is calculated.

Every pass needs $O(sn)$ space from the main memory.
After twice passing over the stream and calculating $\mathcal I^i$s, $\mathcal P^i$s and $z^i$s,
the approximate number of triangles can be determined in $O(1)$ time.
We can store random variables $X_t$, if $\widetilde{\Delta^v}$ is known.
In this case, we will need $O( \frac{\widetilde{\Delta^v} n \log n}{\widehat{\Delta^v} \epsilon^2}  )$ space
to provide a $1 \pm \epsilon $ approximation. Therefore, we can present the following theorem:
\begin{theorem}
\label{theorem:stream}
There is a 2-pass algorithm (if $n$ is already known, otherwise it will need 3 passes)
to count the number of triangles in a stream of edges which needs  
$O(sn)$ memory space and constant update time and its error guarantee obeys Equation \ref{eq:psemivarianceuniform}.
Specifically, with space usage $O( \frac{\widetilde{\Delta^v} n \log n}{\widehat{\Delta^v} \epsilon^2}  )$,
it gives a $1 \pm \epsilon $ approximation.
\end{theorem}

Similar results can be presented for other samplings.

\section{Related work}
\label{sec:relatedwork}

Buriol et al. \cite{buriol:proceedings} presented one of the first approximate triangle counting algorithms.
In their method, an edge and a vertex are selected by random and it is checked whether they form a triangle or not.
Then, the fraction of tests which form a triangle is scaled and returned as an estimation of the number of triangles.
In this method, if
\begin{equation*}
s=\mathbf{log}(\frac{1}{\delta}) \frac{1}{\epsilon^2} \left( \frac{T_0+T_1+T_2+T_3}{T_3} \right)  
\end{equation*}
independent trials are done, where $T_i$ is the number of triples of vertices with $i$ edges,
with probability at least $1-\delta$, the estimated number of triangles is between $(1-\epsilon)T_3$ and $(1+\epsilon)T_3$.
However, $T_0+T_1+T_2$ can be very large compared to $T_3$.

The other method which is efficient only for triangle-dense graphs is \cite{yossef:soda}.
Their method is based on reducing triangle counting to estimating the zero-th, first and second
frequency moments.
They showed that there is a streaming algorithm that for any adjacency stream of a graph,
computes an $(\epsilon-\delta)$ approximation of the number of triangles using space:
\begin{equation*}
O\left( \frac{1}{\epsilon^3} \times \mathbf{log}\frac{1}{\delta}\times\left( \frac{T_1+T_2+T_3}{T_3} \right)^3 \times \textbf{log} n \right) 
\end{equation*}

Tsourakakis \cite{tsourakakis:icdm} proposed a triangle counting algorithm based on
computing the largest eigenvalues of the adjacency matrix of an undirected graph
to approximate both the total as well as the local number
of triangles in the graph.
Tsourakis's method exploits the following property:
the total number of triangles in an undirected graph is
$
\frac{1}{6} \sum_{j=1}^n {\lambda_j}^3
$,
where $\lambda_j$ is the $j$-th eigenvalue of $A$.
Using the Lanczos method \cite{cullum:book}, the eigenvalues are generated from the biggest one to the smallest one.
An approximation is done when the smallest generated eigenvalue contributes very little to the total number of triangles.
However, it is known that time complexity of finding eigenvalues is the same as time complexity of matrix multiplication.
For approximating the number of local triangles of a vertex $i$, he exploited the following property:
$\Delta_i=\frac{\sum_{j=1}^n {\lambda_j}^3 {u_{i,j}}^2 }{2}$,
where $\lambda_j$ is the $j$-th eigenvalue and $u_{i,j}$ is the $j$-th entry of the $i$-th eigenvector of $A$.

In \cite{tsourakakis:proceedings}, the authors proposed the \textsc{DOULION} algorithm for approximate triangle counting.
In this method, every edge $e$ of the graph is removed with a sparsification probability $1-p$.
If it survives, the weight $\frac{1}{p}$ is assigned to it.
Then, the number of triangles in the original graph is approximated
as the number of triangles in the sparsified graph times $\frac{1}{p^3}$.
The following error bound was provided for this method:
\begin{equation*}
\mathbf{Var}(\Delta) = \frac{\Delta(p^3-p^6)+2k(p^5-p^6)}{p^6} 
\end{equation*}
where $k$ is the number of pairs of triangles which are edge-disjoint.


The randomized algorithm of \cite{pagh:jrnl} colors the vertices of the graph with $N=\frac{1}{p}$ colors uniformly at random,
counts triangles whose vertices have the same color, and scales that count appropriately.
The authors showed that for enough large values of $p$, their estimation of the number of triangles is concentrated around its expectation.
In \cite{kolountzakis:jrnl}, the authors combined the sparsification techniques and the the idea of vertex partitioning
(into \textit{high degrees} and \textit{low degrees}) presented in \cite{alon:jrnl}.
They developed an $O(m+\frac{ m^{\frac{3}{2}} \log n }{\Delta \epsilon^2})$ time
$(1\pm \epsilon)$-approximation algorithm.

In \cite{becchetti:jrnl}, the authors studied the problem of approximate local triangle counting in large graphs.
Their approximation algorithms are based on the idea of min-wise independent permutations \cite{Broder:proceedings}.
Their algorithms operate in a semi-streaming fashion, using $O(n)$ space in main memory and performing 
$O(\mathbf{log} n)$ passes over the edges of the graph.

\section{Conclusions}
\label{sec:conclusions}

In this paper, we proposed a new randomized algorithm for approximate triangle counting.
The algorithm provides a general framework
which can be adopted with different sampling techniques and give methods with different time complexities and error bounds.
For example, it can be adopted with the \textit{$q$-optimal sampling} and the \textit{edge sampling}, presented in the paper,
and give linear time algorithms for approximate triangle counting. 
We showed that if an upper bound $\widetilde{\Delta^e}$ is known for the number of triangles incident to every edge,
the proposed method provides an $1\pm \epsilon$ approximation
which runs in $O( \frac{\widetilde{\Delta^e} n \log n}{\widehat{\Delta^e} \epsilon^2}  )$ time,
where $\widehat{\Delta^e}$ is the average number of triangles incident to an edge.
Also, if an upper bound $\widetilde{\Delta^v}$ is known for the number of triangles incident to every vertex,
the proposed method provides an $1\pm \epsilon$ approximation
which runs in $O( \frac{\widetilde{\Delta^v} m \log n}{\widehat{\Delta^v} \epsilon^2}  )$ time,
where $\widehat{\Delta^v}$ is the average number of triangles incident to a vertex.

Finally we showed the algorithm can be extended to streams.
Then it, for example, will perform 2 passes over the data (if the size of the graph is known, otherwise it needs 3 passes)
and will use $O(sn)$ space.  


\bibliographystyle{plain}
\bibliography{approximate}

\begin{thebibliography}{10}

\bibitem{alon:jrnl}
Noga Alon, Raphael Yuster, and Uri Zwick.
\newblock Finding and counting given length cycles.
\newblock {\em Algorithmica}, 17(3):209--223, 1997.

\bibitem{yossef:soda}
Ziv Bar-Yossef, Ravi Kumar, and D.~Sivakumar.
\newblock Reductions in streaming algorithms, with an application to counting
  triangles in graphs.
\newblock In {\em SODA}, pages 623--632. ACM, 2002.

\bibitem{becchetti:jrnl}
Luca Becchetti, Paolo Boldi, Carlos Castillo, and Aristides Gionis.
\newblock Efficient algorithms for large-scale local triangle counting.
\newblock {\em ACM Transactions on Knowledge Discovery from Data (TKDD)},
  4(3):1--28, 2010.

\bibitem{Broder:proceedings}
A.~Z. Broder.
\newblock On the resemblance and containment of documents.
\newblock In {\em In Proceedings of the Compression and Complexity of
  Sequences}, pages 21--29. IEEE, 1998.

\bibitem{buriol:proceedings}
Luciana~S. Buriol, Gereon Frahling, Stefano Leonardi, Alberto
  Marchetti-Spaccamela, and Christian Sohler.
\newblock Counting triangles in data streams.
\newblock In {\em In Proceedings of the twenty-fifth ACM SIGMOD-SIGACT-SIGART
  symposium on Principles of database systems (PODS)}, pages 253--262. ACM,
  2006.

\bibitem{coppersmith:jrnl}
D.~Coppersmith and S.~Winograd.
\newblock Matrix multiplication via arithmetic progressions.
\newblock {\em J. Symbolic Computation}, 9(3):251--280, 1990.

\bibitem{cullum:book}
Jane~K. Cullum and Ralph~A. Willoughby.
\newblock {\em Lanczos Algorithms for Large Symmetric Eigenvalue Computations,
  Vol. 1: Theory}.
\newblock Classics in Applied Mathematics 41, SIAM, 2002.

\bibitem{drineas:soda}
Petros Drineas, Alan~M. Frieze, Ravi Kannan, Santosh Vempala, and V.~Vinay.
\newblock Clustering in large graphs and matrices.
\newblock In {\em SODA}, pages 291--299. ACM, 1999.

\bibitem{drineas:jrnl}
Petros Drineas, Ravi Kannan, and Michael~W. Mahoney.
\newblock Fast monte carlo algorithms for matrices i: Approximating matrix
  multiplication.
\newblock {\em SIAM J. Comput.}, 36(1):132--157, 2006.

\bibitem{drineas2:jrnl}
Petros Drineas, Ravi Kannan, and Michael~W. Mahoney.
\newblock Fast monte carlo algorithms for matrices ii: Computing a low-rank
  approximation to a matrix.
\newblock {\em SIAM J. Comput.}, 36(1):158--183, 2006.

\bibitem{harary:jrnl}
Frank Harary and Helene~J. Kommel.
\newblock Matrix measures for transitivity and balance.
\newblock {\em Journal of Mathematical Sociology}, 6:199--210, 1979.

\bibitem{jowhari:proceedings}
Hossein Jowhari and Mohammad Ghodsi.
\newblock New streaming algorithms for counting triangles in graphs.
\newblock In {\em COCOON}, pages 710--716. Springer, 2005.

\bibitem{kolountzakis:jrnl}
Mihail~N. Kolountzakis, Gary~L. Miller, Richard Peng, and Charalampos~E.
  Tsourakakis.
\newblock Efficient triangle counting in large graphs via degree-based vertex
  partitioning.
\newblock {\em Internet Mathematics}, 8(1-2):161--185, 2012.

\bibitem{pagh:jrnl}
Rasmus Pagh and Charalampos~E. Tsourakakis.
\newblock Colorful triangle counting and a mapreduce implementation.
\newblock {\em Inf. Process. Lett.}, 112(7):277--281, 2012.

\bibitem{tsourakakis:icdm}
Charalampos~E. Tsourakakis.
\newblock Fast counting of triangles in large real networks without counting:
  Algorithms and laws.
\newblock In {\em In Proceedings of the IEEE International Conference on Data
  Mining (ICDM)}, pages 608--617. IEEE, 2008.

\bibitem{tsourakakis:proceedings}
Charalampos~E. Tsourakakis, U.~Kang, Gary~L. Miller, and Christos Faloutsos.
\newblock Doulion: counting triangles in massive graphs with a coin.
\newblock In {\em KDD}, pages 837--846, 2009.

\bibitem{tsourakakis:jrnl}
Charalampos~E. Tsourakakis, Mihail~N. Kolountzakis, and Gary~L. Miller.
\newblock Triangle sparsifiers.
\newblock {\em J. Graph Algorithms Appl.}, 15(6):703--726, 2011.

\bibitem{watts:jrnl}
Duncan~J. Watts and Steven~H. Strogatz.
\newblock Dynamics of small-world networks.
\newblock {\em Nature}, 393:440--442, 1998.

\bibitem{williams:proceedings}
Virginia~Vassilevska Williams.
\newblock Multiplying matrices faster than coppersmith-winograd.
\newblock In {\em STOC}, pages 887--898. ACM, 2012.

\end{thebibliography}

\end{document}